\documentclass[letterpaper, 10pt, conference]{ieeeconf}  

\IEEEoverridecommandlockouts                              

\title{\LARGE\bf Covariance Dynamics and Entanglement in Translation Invariant Linear Quantum Stochastic Networks} 
\author{ Arash Kh. Sichani, \qquad Igor G. Vladimirov, \qquad Ian R. Petersen
\thanks{This work is supported by the Australian Research Council. The authors are with UNSW Canberra, ACT 2600, Australia. E-mail: {\tt arash\_kho@hotmail.com, igor.g.vladimirov@gmail.com, i.r.petersen@gmail.com}.}
}

\usepackage{mathptmx} 

\usepackage{graphicx} 
\usepackage{times} 
\usepackage{amsmath} 
\usepackage{amssymb}  
\usepackage{datetime}
\usepackage{comment}
\usepackage{tikz}
\usetikzlibrary{calc}
\usetikzlibrary{matrix}
\newtheorem{lem}{Lemma}
\newtheorem{thm}{Theorem}

\def\<{\leqslant}           
\def\>{\geqslant}           

\def\wh{\widehat}
\def\wt{\widetilde}
\def\mod{\mathrm{mod}}   
\def\Re{\mathrm{Re}}   
\def\Im{\mathrm{Im}}   

\def\mA{\mathbb{A}}    
\def\mR{\mathbb{R}}    
\def\mC{\mathbb{C}}    

\def\Tr{\mathrm{Tr}}       
\def\rT{\mathrm{T}}        

\def\bE{\mathbf{E}}    

\def\[[[{[\![\![}
\def\]]]{]\!]\!]}

\def\re{\mathrm{e}}        
\def\rd{\mathrm{d}}        

\def\br{\mathbf{r}}
\def\x{\times}
\def\ox{\otimes}

\def\vec{\mathrm{vec}}

\def\cW{\mathcal{W}}

\def\cX{\mathcal{X}}

\def\cA{\mathcal{A}}

\def\cE{\mathcal{E}}

\def\cS{\mathcal{S}}

\def\mU{\mathbb{U}}

\def\mH{{\mathbb H}}
\def\mS{{\mathbb S}}

\def\Ups{\Upsilon}

\DeclareMathAlphabet{\mathbfit}{OML}{cmm}{b}{it}

\onecolumn
\begin{document}
\maketitle
\thispagestyle{empty}

\begin{abstract}
This paper is concerned with a translation invariant network of identical quantum stochastic systems subjected to external quantum noise. Each node of the network is directly coupled to a finite number of its neighbours. This network is modelled as an open quantum harmonic oscillator and is governed by a set of linear quantum stochastic differential equations (QSDEs). The dynamic variables of the network satisfy the canonical commutation relations (CCRs). Similar large-scale networks can be found, for example, in quantum metamaterials and optical lattices. Using spatial Fourier transform techniques, we obtain a sufficient condition for stability of the network  in the case of finite interaction range, and consider a mean square  performance index for the stable network in the thermodynamic limit. The Peres-Horodecki-Simon separability criterion is employed in order to obtain sufficient and necessary conditions for quantum entanglement of bipartite systems of nodes of the network in the Gaussian invariant state. The results on stability and entanglement are extended to the infinite chain of the linear quantum systems by letting the number of nodes go to infinity. A numerical example is provided to illustrate the results.
\end{abstract}
\section{INTRODUCTION}
Recently, there has been an increased interest in the study of dynamics of large-scale quantum networks. In particular, one can refer to the studies on artificial optical media known as quantum metamaterials; see, for example, \cite{Rakhmanov2008, Notomi2008, Zheludev11, Quach11, Zheludev11, Zagoskin12}. These networks are composed of complex unit cells and are effectively homogeneous on the scale of relevant wavelengths, for example, in the microwave range. In comparison with natural solids, whose optical behaviour is specified by the quantum energy level configurations of the constituent atoms or molecules, the electromagnetic response of quantum metamaterials stems from the controllable resonant characteristics of their building elements such as the Josephson devices or optical cavities \cite{Rakhmanov2008, Notomi2008,Quach11}. 
Quantum metamaterials are considered to be a promising approach to the implementation of quantum computer elements, such as quantum bits (qubits) and qubit registers, which can maintain quantum coherence over many cycles of their internal evolution \cite{Quach11, Zagoskin11}.

The machinery of linear quantum stochastic differential equations (QSDEs) provides a framework for modelling and analysis of a wide range of open quantum systems \cite{HP_1984,P_1992}, including those which arise in quantum optics. In the QSDEs,  the quantum noise   from the surroundings is modelled as a heat bath of external fields acting on a boson Fock space \cite{P_1992}. The class of linear QSDEs represents the Heisenberg evolution of pairs of conjugate operators in a multi-mode open quantum harmonic oscillator which is coupled to external bosonic fields. In this framework,  the modelling of translation invariant networks with field coupling between the subsystems and stability analysis for certain classes of nonlinear open quantum systems have been addressed, for example, in  \cite{VP2014,Ian2012robust,SVP_2015}.

For modelling purposes, the analysis of homogeneous large-scale networks can be simplified by using periodic boundary conditions (PBCs) \cite{Quach11, VP2014}. 
The PBCs rely on negligibility of boundary effects in a fragment of a translation invariant network consisting of a sufficiently large number of subsystems. This approximation technique is used for lattice models of interacting particle systems in statistical physics, including, for example, the Ising model of ferromagnetism \cite{Newman99}.

Quantum decoherence, as the loss of quantum nature  of the system state towards classical states,  is described by using entanglement measures and correlations between the nodes of bipartite systems \cite{Serafini05}. Among purely quantum mechanical phenomena,  entanglement (or inseparability) plays a crucial role in the fields of  quantum information and quantum computation \cite{NC_2000,Ho2009}. Various measures have been presented for entanglement; see, for example, \cite{Ho2009}. A criterion for checking whether a given system state is separable was proposed in \cite{Peres96}. This condition was shown to be necessary and sufficient for separability of Gaussian states  in two-mode quantum harmonic oscillators \cite{Sim2000}.

In this paper, we consider large fragments of a  translation invariant network of identical linear quantum stochastic systems endowed with the PBCs. The nodes of the network are directly coupled  to each other   (through bilateral energy transfer) within a finite interaction range. This network is modelled by linear QSDEs based on Hamiltonian and coupling parametrization of the corresponding multi-mode open quantum harmonic oscillator. The dynamic variables of the network satisfy the canonical commutation relations (CCRs). Using spatial Fourier transforms \cite{VP2014} and linear matrix inequalities (LMIs), we present sufficient conditions for stability of the network, thus extending the previously considered nearest neighbour interaction setting to the case  when each node of the network is coupled to a finite number of its neighbours. We also consider a mean square performance index for the stable network in the thermodynamic limit. In the case when the network consists of one-mode oscillators, we obtain sufficient and necessary conditions for quantum entanglement between nodes of the network in the Gaussian invariant state, based on the Peres-Horodecki-Simon separability criteria. By letting the number of nodes go to infinity, the results on stability and entanglement are extended to the infinite network of linear quantum stochastic systems. For the infinite network in the Gaussian invariant state, it is shown that the entanglement disappears for nodes which are sufficiently far away from each other. However, we provide a numerical example of a large sample of long fragments of a stable network with randomly generated parameters, where the bipartite entanglement between nodes is present within the interaction range.

The rest of the paper is organized as follows. Section~\ref{sec:not} provides principal notation. Section~\ref{sec:system} specifies the class of open quantum systems being considered. Section~\ref{sec:Ring_Topology} describes a model for a fragment of a translation  invariant network of such systems with PBCs. Section~\ref{sec:Spatial_fourier_tran} employs spatial Fourier transforms in order to obtain a more tractable description for the covariance  dynamics of the network. Section~\ref{sec:Stable_TIN} discusses a stability condition and a mean square performance functional. Section~\ref{sec:Sep_TIN} presents conditions for entanglement between nodes of the translation invariant network of one-mode oscillators. Section~\ref{sec:inf_chain} extends the results on stability and entanglement to the infinite chain of such systems. Section~\ref{sec:example} presents a numerical example to illustrate the results. Section~\ref{sec:conc} makes concluding remarks. Appendices~\ref{app:Lya_sol}, \ref{app:fact_on_ent} and \ref{app:fact_about_matrices} provide subsidiary material on a closed-form solution of the Lyapunov equation and entanglement criterion of bipartite systems, and some facts about matrices, respectively.

\section{NOTATION}\label{sec:not}

Unless specified otherwise,  vectors are organized as columns, and the transpose $(\cdot)^{\rT}$ acts on matrices with operator-valued entries as if the latter were scalars. For a vector $X$ of operators $X_1, \ldots, X_r$ and a vector $Y$ of operators $Y_1, \ldots, Y_s$, the commutator matrix is defined as an $(r\x s)$-matrix
$
    [X,Y^{\rT}]
    :=
    XY^{\rT} - (YX^{\rT})^{\rT}
$
whose $(j,k)$th entry is the commutator
$
    [X_j,Y_k]
    :=
    X_jY_k - Y_kX_j
$ of the operators $X_j$ and $Y_k$.
Also, $(\cdot)^{\dagger}:= ((\cdot)^{\#})^{\rT}$ denotes the transpose of the entry-wise operator adjoint $(\cdot)^{\#}$. In application to complex matrices,  $(\cdot)^{\dagger}$ reduces to the complex conjugate transpose  $(\cdot)^*:= (\overline{(\cdot)})^{\rT}$. Furthermore, $\mS_r$, $\mA_r$
 and
$
    \mH_r
    :=
    \mS_r + i \mA_r
$ denote
the subspaces of real symmetric, real antisymmetric and complex Hermitian  matrices of order $r$, respectively, with $i:= \sqrt{-1}$ the imaginary unit. Also, $I_r$ denotes the identity matrix of order $r$, positive (semi-) definiteness of matrices is denoted by ($\succcurlyeq$) $\succ$, and $\ox$ is the tensor product of spaces or operators (for example, the Kronecker product of matrices). The generalized (Moore-Penrose) inverse of a matrix $M$ is denoted by $M^+$, and $\br(M)$ denotes the spectral radius of $M$. 
The Kronecker delta is denoted by $\delta_{jk}$, and $\mU:= \{z\in \mC:\ |z|=1\}$ is the unit circle in the complex plain. The floor function is denoted by $\lfloor \cdot \rfloor$.
Also, $\bE \xi := \Tr(\rho \xi)$ denotes the quantum expectation of a quantum variable $\xi$ (or a matrix of such variables) over a density operator $\rho$ which specifies the underlying quantum state. For matrices of quantum variables, the expectation is evaluated entry-wise.

\section{LINEAR QUANTUM STOCHASTIC SYSTEMS}\label{sec:system}

We consider a quantum stochastic system interacting with external boson fields.  The system has $N$ subsystems with associated $n$-dimensional vectors $X_0, \ldots, X_{N-1}$ of dynamic variables which satisfy the CCRs
\begin{equation}
\label{xCCR}
    [X, X^{\rT}] = 2i \mathbfit{\Theta},
    \qquad
    \mathbfit{\Theta}:= I_N \ox \Theta,
    \qquad
    X:=
    {\small\begin{bmatrix}
        X_0\\
        \vdots\\
        X_{N-1}
    \end{bmatrix}}.
\end{equation}
Here, $\mathbfit{\Theta}$ is a block diagonal joint CCR matrix, where $\Theta \in \mA_{2n}$ is a nonsingular matrix which is usually of the form $\Theta = I_n \ox {\scriptsize\begin{bmatrix}0 & 1\\ -1 & 0\end{bmatrix}}$. The system variables evolve in time according to the QSDE
\begin{equation}
\label{dx}
  \rd X
  \!=\!
  \Big(
    i[H,X] \!-\!\frac{1}{2} \mathbfit{B}J\mathbfit{B}^{\rT} \mathbfit{\Theta}^{-1} X
  \Big)\rd t
  \!+\!
  \mathbfit{B} \rd W,
  \quad
  W \!:=\!\!
  {\small\begin{bmatrix}
  W_0\\
  \vdots \\
  W_{N-1}
  \end{bmatrix}}.\!\!
\end{equation}
Here, $W_0, \ldots, W_{N-1}$ are $2m$-dimensional vectors of quantum Wiener processes  with a  positive semi-definite It\^{o} matrix $\Omega \in \mH_{2m}$:
\begin{equation}
\label{WW}
    \rd W_j \rd W_k^{\rT}
    =
    \delta_{jk}\Omega \rd t,
    \qquad
    \Omega := I_{2m} + iJ,
\end{equation}
where $J := I_m \ox {\scriptsize\begin{bmatrix}
0 & 1\\ -1 & 0
\end{bmatrix}}$. The matrix $\mathbfit{B} \in \mR^{2nN\x 2mN}$ in (\ref{dx}) is related to a matrix $\mathbfit{M} \in \mR^{2mN \x 2nN}$ of linear dependence of the system-field coupling operators on the system variables by
\begin{equation}
\label{BM}
    \mathbfit{B} = 2\mathbfit{\Theta} \mathbfit{M}^{\rT}.
\end{equation}
The term $-\frac{1}{2} \mathbfit{B}J\mathbfit{B}^{\rT} \mathbfit{\Theta}^{-1} X$ in the drift of the QSDE (\ref{dx}) is the Gorini-Kossakowski-Sudarshan-Lindblad decoherence superoperator \cite{GKS_1976,L_1976} which acts on the system variables and is associated with the system-field interaction. Also, $H$ is the Hamiltonian which describes the self-energy of the system and is usually represented as a function of the system variables. 
In the case of an open quantum harmonic oscillator \cite{EB_2005,GZ_2004}, the Hamiltonian $H$ is a quadratic function of the system variables
\begin{equation}
\label{H0}
    H
    :=
    \frac{1}{2}
    X^{\rT} R X
    =
    \frac{1}{2}
    \sum_{j,k=0}^{N-1}
    X_j^\rT R_{jk} X_k,
\end{equation}
where the matrix $R:= (R_{jk})_{0\< j,k< N} \in \mS_{2nN}$ is formed from $(2n\x 2n)$-blocks $R_{jk}$. 
By substituting 
(\ref{H0}) into (\ref{dx}) and using the CCRs (\ref{xCCR}), it follows that the QSDE takes the form of a linear QSDE
\begin{equation}
\label{dx1}
  \rd X = A X\rd t+ \mathbfit{B} \rd W,
\end{equation}
where the system matrix $A \in \mR^{2nN\x 2nN}$ is given by
\begin{equation}\label{A}
    A
    :=
    2\mathbfit{\Theta} R - \frac{1}{2} \mathbfit{B}J\mathbfit{B}^{\rT} \mathbfit{\Theta}^{-1} .
\end{equation}

\section{LINEAR QUANTUM STOCHASTIC NETWORK WITH PERIODIC BOUNDARY CONDITIONS}\label{sec:Ring_Topology}

Suppose the open quantum harmonic oscillator of the previous section is a fragment of a translation invariant network which is organised as a one-dimensional chain of identical linear quantum stochastic  systems numbered by $k=0,\ldots, N-1$. Each node of the network interacts with the corresponding external boson field and hence, the joint network-field coupling matrix $\mathbfit{M}$ in (\ref{BM}) is block diagonal:
\begin{equation}
\label{MM}
  \mathbfit{M} = I_N \ox M,
\end{equation}
where $M \in \mR^{m\x n}$. Furthermore, the nodes are directly coupled to each other within a finite interaction range $d$. The fragment of the chain is assumed to be  large enough in the sense that $N > 2d$, and is  endowed with the PBCs, thus having a ring topology. A particular case of nearest neighbour interaction is depicted in Fig.~\ref{fig:Ring_Topology}.
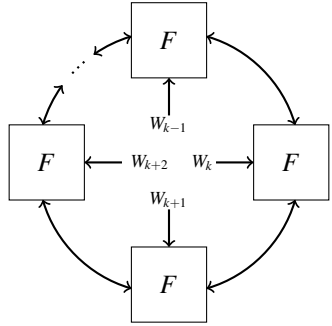
\begin{figure}[htbp]
	\centering
	\begin{tikzpicture}[scale=.5]
		\node at (0,1) {\scriptsize $W_{k-1}$};
		\node at (0.90,0) {\scriptsize $W_k$};
		\node at (0,-1) {\scriptsize $W_{k+1}$};
		\node at (-.5,0) {\scriptsize $W_{k+2}$};

		\draw [->, thick] (1.25,0) -- (2.25,0);
		\draw [->, thick] (-1.25,0) -- (-2.25,0);
		\draw [->, thick] (0,1.25) -- (0,2.25);
		\draw [->, thick] (0,-1.25) -- (0,-2.25);

		\draw  (2.25,1) rectangle (4.25,-1);
		\draw  (-1,4.25) rectangle (1,2.25);
		\draw  (-2.25,-1) rectangle (-4.25,1);
		\draw  (1,-2.25) rectangle (-1,-4.25);

		\node at (3.25,0) {$F$};
		\node at (0,3.25) {$F$};
		\node at (-3.25,0) {$F$};
		\node at (0,-3.25) {$F$};
		\draw [<->, thick] ([shift=(16.5:3.5)] 0,0) arc (16.5:73.5:3.5); 
		\draw [<->, thick] ([shift=(106.5:3.5)] 0,0) arc (106.5:125:3.5); 
		\draw [dotted, thick] ([shift=(130:3.5)] 0,0) arc (130.5:140:3.5);
		\draw [<->, thick] ([shift=(144:3.5)] 0,0) arc (144:163.5:3.5);
		\draw [<->, thick]([shift=(196.5:3.5)] 0,0) arc (196.5:253.5:3.5); 
		\draw [<->, thick]([shift=(-16.5:3.5)] 0,0) arc (-16.5:-73.5:3.5); 
		\end{tikzpicture}
\caption{A finite fragment of a translation invariant network of open quantum systems with direct coupling between the nearest neighbours. Also shown are the quantum noises which drive the nodes of the network.}
\label{fig:Ring_Topology}
\end{figure}
In the case of an arbitrary interaction range $d\>1$,  the Hamiltonian $H$ in (\ref{H0}) is completely specified by matrices $R_{\ell} = R_{-\ell}^{\rT} \in \mR^{2n\x 2n}$, with $\ell = 0, \pm 1, \ldots, \pm d$, as
\begin{equation}
\label{HR}
    H := \frac{1}{2}\sum_{j=0}^{N-1} \Big(X_j^{\rT}\sum_{\ell = -d}^d R_{\ell} X_{\mod(j-\ell, N)}\Big).
\end{equation}
Here, $j-\ell$ is computed modulo $N$ in accordance with the PBCs, and hence, the corresponding matrix $R$ is block circulant. The matrix $R_0 \in \mS_{2n}$ specifies the free Hamiltonian for each node, while $R_{\pm s}$  describe the energy coupling between the nodes which are at a distance $s=1, \ldots, d$ from each other (with the more distant nodes in the network not being directly coupled).
In view of the block diagonal structure of the matrices $\mathbfit{\Theta}$ and $\mathbfit{M}$ in (\ref{xCCR}) and (\ref{MM}), it follows from  (\ref{BM}), (\ref{A}) and (\ref{HR}), that the QSDE (\ref{dx1}) is representable as a set of coupled QSDEs for the dynamic variables of the nodes of the network:
\begin{equation}
	\label{dXj}
  		\rd X_j = \sum_{\ell =-d}^d A_{\ell} X_{\mod(j-\ell,N)}\rd t + B \rd W_j,
    \quad
    j = 0,\ldots, N-1.
\end{equation}
Here, similarly to $R$, the network dynamics matrix $A$ in (\ref{A}) also has a block circulant structure, with
\begin{align}
\label{AA}
    A_{\ell}
    & :=
    \left\{
        \begin{array}{ll}
            2 \Theta R_0 - \frac{1}{2} B J B^\rT \Theta^{-1} & {\rm if}\ \ell = 0\\
            2 \Theta R_{\ell} & {\rm if}\ \ell \ne 0
        \end{array}
    \right.,\\
\label{BB}
    B & := 2 \Theta M^{\rT}.
\end{align}

\section{SPATIAL FOURIER TRANSFORMS AND COVARIANCE DYNAMICS OF THE NETWORK}\label{sec:Spatial_fourier_tran}

The dynamics (\ref{dXj}) of the nodes of the network can be studied in the spatial frequency domain. Similarly to \cite{VP2014}, we will use the discrete Fourier transforms (DFTs) of the quantum processes $X_k$ and $W_k$ over the spatial subscript $k=0, ..., N-1$:
\begin{align}
	\label{equ:ztran:x_k}
	\cX_z(t)&:=\sum_{k=0}^{N-1}z^{-k}X_k(t),\\
	\label{equ:ztran:w_k}
	\cW_z(t)&:=\sum_{k=0}^{N-1}z^{-k}W_k(t),	
\end{align}
where $z\in \mU_N$ and
\begin{equation}
	\label{mUN}
	\mU_N := \big\{ \re^{\frac{ 2\pi i k}{N}} : \ k=0,\ldots,N-1 \big\}
\end{equation}
is the set of $N$th roots of unity. The quantum process  $\cX_z(t)$ is evolved in time by the QSDE
\begin{align}
	\nonumber	
	\rd \cX_z  &= \sum_{j=0}^{N-1}z^{-j}\rd X_j\\
\nonumber
    &= \sum_{j=0}^{N-1}z^{-j}
                \left(
                    \sum_{\ell=-d}^d A_{\ell} X_{\mod(j-\ell,N)}\rd t+ B \rd W_j
                \right)\\
			   &= \cA_z \cX_z \rd t+ B \rd \cW_z,
			   \label{equ:ztran_evolution}			
\end{align}
where the matrix  $A_z$ is defined in terms of the matrices (\ref{AA}) as
\begin{equation}
		\label{equ:A_z}
		\cA_z:= \sum_{k=-d}^{d} A_k z^{-k},
\end{equation}
and $B$ is given by (\ref{BB}).
Here, use is made of (\ref{equ:ztran:x_k}) and (\ref{equ:ztran:w_k}) together with the PBCs and the well-known properties of DFTs due to the assumption that $z\in \mU_N$. Furthermore, it follows from the CCRs (\ref{xCCR}) that
$$	
    [\cX_z, \cX_v^\dagger] = \sum_{j,k=0}^{N-1} z^{-j} v^k [X_j,X_k^\rT]
						=2i\Theta \sum_{k=0}^{N-1}(v/z)^k
						=2i\delta_{zv}N\Theta
$$
for all $z,v\in \mU_N$. By a similar reasoning, (\ref{WW}) implies that
\begin{align}
	\label{equ:Noise_Ito_net}
	\rd \cW_z \rd \cW_v^\dagger = N \delta_{zv} \Omega \rd t.
\end{align}
We will need the following matrix of second-order cross-moments of the quantum processes $\cX_z$ and $\cX_v$ given by (\ref{equ:ztran:x_k}):
	\begin{equation}
		\label{equ:DFT_cov}
		\cS_{z,v}(t)
        :=
        \bE
        (
            \cX_z(t)
            \cX_v(t)^{\dagger}
        ).	
	\end{equation}
Here and in what follows,  the quantum expectation $\bE(\cdot)$ is over the tensor product $\rho:= \varpi \ox \upsilon$  of the initial state $\varpi$ of the network and the vacuum state $\upsilon$ of the external fields.

\begin{lem}
	\label{lem:cov_net_time}
    For any $N$th roots of unity $z,v\in \mU_N$ from (\ref{mUN}), the matrix $\cS_{z,v}(t)$ in (\ref{equ:DFT_cov})
	satisfies the differential Lyapunov equation
	\begin{equation}
		\label{Sdot}
		\dot{\cS}_{z,v} = \cA_z \cS_{z,v} + \cS_{z,v} \cA_v^* + N\delta_{zv} B \Omega B^\rT
	\end{equation}		
for any $t\>0$, 	where the matrix $\cA_z$ is given by (\ref{equ:A_z}). 	Furthermore, suppose $\cA_z$ is Hurwitz for all $z\in \mU_N$. Then, for any $z, v \in \mU_N$, there exists a limit
\begin{equation}
\label{Szvinfdef}
    \cS_{z,v}(\infty):= \lim_{t\to +\infty} \cS_{z,v}(t)
\end{equation}
which is a unique solution of the algebraic Lyapunov equation
	\begin{equation}
		\label{equ:cov_inf_net}
		\cA_z \cS_{z,v}(\infty) + \cS_{z,v}(\infty) \cA_v^* + N\delta_{zv} B \Omega B^\rT=0.
	\end{equation}
\end{lem}	
\begin{proof}
A combination of the QSDE (\ref{equ:ztran_evolution}) with the quantum It{\^ o} formula leads to
	\begin{align*}
		\rd \cS_{z,v} &= \rd \bE(\cX_z\cX_v^\dagger)\\
					  & = \bE(\rd \cX_z\cX_v^\dagger + \cX_z \rd \cX_v^\dagger
					  +
					  \rd \cX_z \rd \cX_v^\dagger)\\
					  & = (\cA_z \cS_{z,v} + \cS_{z,v} \cA_v^* + N\delta_{zv} B \Omega B^\rT)\rd t,
	\end{align*}
which implies the ODE (\ref{Sdot}). Here, use is made of (\ref{equ:Noise_Ito_net}) and the fact that the forward increments of the quantum Wiener process in the vacuum state are uncorrelated with the adapted processes, and hence, $\bE(\cX_z \rd \cW_v^{\dagger})=0$.
The solution of (\ref{Sdot}) is described by
\begin{equation}
\label{sol}
		\cS_{z,v}(t)
        =
        \re^{t\cA_z} \cS_{z,v}(0) \re^{t\cA_v^*}
        +
        N\delta_{zv}
        \int_0^t
        \re^{\tau \cA_z}
        B \Omega B^\rT
        \re^{\tau \cA_v^*}\rd \tau\!
\end{equation}
for any $t\>0$. Therefore, under the assumption that $\cA_z$ is Hurwitz for all $z\in \mU_N$,  the solution (\ref{sol}) has the following limit (\ref{Szvinfdef}):
$$		
    \cS_{z,v}(\infty)
        =
        N\delta_{zv}
        \int_0^{+\infty}
        \re^{\tau \cA_z}
        B \Omega B^\rT
        \re^{\tau \cA_v^*}
        \rd \tau
$$
	which is a unique steady-state solution of (\ref{Sdot}) reducing to (\ref{equ:cov_inf_net}).
\end{proof}

Under the assumptions of Lemma~\ref{lem:cov_net_time}, the solution of (\ref{equ:cov_inf_net}) is representable as
\begin{equation}
\label{Szvinf}
    \cS_{z,v}(\infty) = N \delta_{zv} \cS_z,
\end{equation}
for all $z,v\in \mU_N$,
where $\cS_z=\cS_z^*\succcurlyeq 0$ is a unique solution of the algebraic Lyapunov equation
\begin{equation}
\label{cSz}
    \cA_z \cS_{z} + \cS_{z} \cA_z^* + B \Omega B^\rT=0.
\end{equation}
The function $\mU \ni z\mapsto S_z$ is the spatial spectral density \cite{VP2014} which encodes the covariance structure of the dynamic variables of the network in the invariant Gaussian quantum state (in the limit of infinite time and infinite network size). Also note that the cross-moments of the dynamic variables at the $j$th and $k$th nodes of the network can be recovered from the matrices (\ref{equ:DFT_cov}) by applying the inverse DFT to (\ref{equ:ztran:x_k}):
	\begin{align}
    \nonumber
		\bE(X_j(t) X_k(t)^\rT) &= \frac{1}{N^2} \bE \sum_{z,v \in \mU_N} z^j v^{-k}\cX_z(t)\cX_v(t)^\dagger \\
    \label{equ:cros_cov}
						 &= \frac{1}{N^2} \sum_{z,v \in \mU_N} z^j v^{-k}\cS_{z,v}(t)
	\end{align}		
for all $0\< j,k< N$.
In fact, the right-hand side of (\ref{equ:cros_cov}) is the two-dimensional inverse DFT of the matrices $\cS_{z,1/v}$.

\section{STABILITY AND MEAN SQUARE PERFORMANCE OF THE NETWORK}\label{sec:Stable_TIN}

Recall that the network matrix $A$ in (\ref{A}), which corresponds to the set of QSDEs (\ref{dXj}),  has a block circulant structure specified by the matrices (\ref{AA}). Therefore, $A$ is Hurwitz if and only if the matrix $\cA_z$ in (\ref{equ:A_z}) is Hurwitz for all $z \in \mU_N$. The latter property holds for any fragment length $N\> 1$ if and only if
\begin{equation}
\label{stab}
	\max_{z \in \mU} \br(\text{e}^{\cA_z}) < 1
\end{equation}
(since $\cA_z$ is a continuous function of $z$ over the unit circle, and the set $\bigcup_{N=1}^{+\infty}\mU_N$ is dense in $\mU$). In what follows, the translation invariant network, described in Section~\ref{sec:Ring_Topology}, is called stable if it satisfies (\ref{stab}).
The following lemma provides a sufficient condition for stability of the network under consideration.
\begin{lem}
	\label{lem:stab}
Suppose there exist positive definite matrices $\cS, Q\in \mS_{2n}$ satisfying the LMI
\begin{equation}
	\label{equ:LMI}
	{\small
	\begin{bmatrix}
		A_0 \cS + \cS A_0^\rT+Q & \cS & \wt{A}\\			
		\cS                         & -I_{2n}      &   0\\
		\wt{A}^\rT                   &  0      &   -\frac{1}{2d}I_{2nd}
	\end{bmatrix}} \preccurlyeq 0,
\end{equation}			
where $\wt{A}\in \mR^{2n \x 4nd}$ is an auxiliary matrix defined in terms of (\ref{AA}) by
\begin{equation}
\label{AAAA}
    \wt{A}:= {\begin{bmatrix} A_{-d} & \ldots & A_{-1} & A_{1} & \ldots & A_{d} \end{bmatrix}}.
\end{equation}
Then the quantum network being considered is stable in the sense of (\ref{stab}).
\end{lem}
\begin{proof}
	In view of the Schur complement lemma \cite{zhang2006} and (\ref{AAAA}), the LMI (\ref{equ:LMI}) implies that
\begin{equation}
\label{Schur}
			A_0 \cS + \cS A_0^\rT + Q + \cS^2 + 2d\sum_{j \neq 0} A_j A_j^\rT \preccurlyeq
			0.
\end{equation}
For any $z\in \mU$, the vector $\Delta:={\begin{bmatrix} z^{d},
	\ldots, z, z^{-1}, \ldots, z^{-d}\end{bmatrix}}^{\rT}\in \mC^{2d}$ satisfies $|\Delta|^2 \< 2d$, and hence, $\Delta\Delta^* \preccurlyeq 2dI_{2d}$. Therefore, application of the completing-the-square technique leads to
	\begin{equation}
    \label{square}
		\sum_{j\neq 0} A_j z^{-j} \cS
			+ \cS \sum_{j \neq 0} A_j^\rT z^{j}
            \preccurlyeq
			2d\sum_{j \neq 0} A_j A_j^\rT + \cS^2,
	\end{equation}
where use is made of the relations $\sum_{j\neq 0} A_j z^{-j} = \wt{A}(\Delta\ox I_{2n})$ and $\wt{A}(\Delta\ox I_{2n})(\wt{A}(\Delta\ox I_{2n}))^* = \wt{A}((\Delta \Delta^*)\ox I_{2n})\wt{A}^{\rT}\preccurlyeq 2d \wt{A}\wt{A}^{\rT}$. A combination of (\ref{Schur}) and (\ref{square}) implies that
		\begin{equation*}
		A_0 \cS + \cS A_0^\rT + \sum_{j \neq 0} A_j z^{-j} \cS
			+ \cS \sum_{j \neq 0} A_j^\rT z^{j} +Q \preccurlyeq 0,
		\end{equation*}
	which is equivalent to
	\begin{align}
			\label{equ:Lyp_in}
			\small
			\cA_z \cS + \cS \cA_z^* + Q \preccurlyeq 0.
	\end{align}
	Since $\cS$ and $Q$ in (\ref{equ:LMI}) were assumed to be positive definite, (\ref{equ:Lyp_in}) implies that $\cA_z$ is Hurwitz for all $z
	\in \mU$, which proves (\ref{stab}).
\end{proof}	

Note that 	the condition for $A_0$ to be Hurwitz is a prerequisite for feasibility of the LMI (\ref{equ:LMI}). 
Now, assuming that the stability condition (\ref{stab}) is satisfied, we will consider the following mean square performance measure \cite{VP2014} for the network of size $N$ at time $t\> 0$:
\begin{equation}
	\label{equ:cost_fun}
	\cE_N(t)
    :=
    \frac{1}{N}
    \bE
    \sum_{j,k=0}^{N-1}
    X_j(t)^\rT \sigma_{j-k} X_k(t).
\end{equation}
Here, $\sigma_k$ is a given $\mR^{2n \x 2n}$-valued sequence which satisfies $\sigma_{−-k} = \sigma_k^\rT$
for all integers $k$ and specifies a real symmetric block Toeplitz weighting matrix $(\sigma_{j-k})_{0\<j,k<N}$. The block Toeplitz structure of the weighting matrix in (\ref{equ:cost_fun}) corresponds to the translation invariance of the quantum network being considered.
A matrix-valued map $\mU \ni z \mapsto \Sigma_z = \Sigma_z^*$, defined by
\begin{equation}
	\label{Sigma_z}
	\Sigma_z:=\sum_{k=-\infty}^{+\infty} z^{-k} \sigma_k,
\end{equation}
describes the spectral density of the weighting sequence. In order to ensure the absolute convergence of the series in (\ref{Sigma_z}), it is assumed that $\sum^{+\infty}_{k=-\infty} \|\sigma_k\| < +\infty$, which also makes $\Sigma_z$ a continuous function of $z$. The fulfillment of the condition $\Sigma_z \succcurlyeq 0$ for all $z \in \mU$ is necessary and sufficient for $(\sigma_{j-k})_{0\<j,k<N} \succcurlyeq 0$ to hold  for all $N\>1$; see, for example \cite{gren58}. In this case, the sum on the right-hand side of (\ref{equ:cost_fun}) is a positive semi-definite operator, and hence, $\cE_N \>0$.

\begin{lem}
\label{lem:cost_fun_inf}
Suppose the stability condition (\ref{stab}) is satisfied. Then for any given length $N$ of the network fragment,  the cost functional in (\ref{equ:cost_fun})  has the following infinite time horizon limit
\begin{equation}	
\label{equ:cost_fun:inf}
	\cE_N(\infty)
    :=\lim_{t\to +\infty} \cE_N(t)
		 =
    \frac{1}{N}
    \sum_{z \in \mU_N}
    \Tr(\wh{\Sigma}_N(z)\cS_z).
\end{equation}
Here,
\begin{equation}
	\label{equ:spec_dens}
	\wh{\Sigma}_N(z):=\sum_{k=1-N}^{N-1}\Big(1-\frac{|k|}{N}\Big)z^{-k}\sigma_k,
\end{equation}
and the matrix $\cS_z$ is the unique solution of the algebraic Lyupunov equation (\ref{cSz}).
\end{lem}
\begin{proof}
The relation (\ref{equ:cros_cov}) allows the expectation on the right-hand side of (\ref{equ:cost_fun}) to be represented in the spatial frequency domain as
	\begin{align}
		\label{equ:lem:weighted_cov}
		\nonumber
		\bE(X_j(t)^\rT \sigma_{j-k}X_k(t))&=\Tr(\sigma_{j-k}^\rT \bE(X_j(t)X_k(t)^\rT))\\
							        	  &=\frac{1}{N^2}\sum_{z,v \in \mU_N} z^jv^{-k}
							        	  \Tr(\sigma_{k-j}\cS_{z,v}(t)), \!\!
	\end{align}
	where use is also made of the property $\sigma_k^\rT=\sigma_{-k}$ of the weighting sequence. Hence, in view of	(\ref{Szvinfdef}), (\ref{Szvinf}) and (\ref{equ:spec_dens}), the infinite time horizon limit of $\cE_N$  takes the form
	\begin{align*}
		\cE_N(\infty)
        &=
        \frac{1}{N^3}
        \sum_{j,k=0}^{N-1}\,
        \sum_{z,v \in \mU_N}
        z^jv^{-k}
		\Tr(\sigma_{k-j}\cS_{z,v}(\infty))\\
		&=
        \frac{1}{N^2}
        \sum_{j,k=0}^{N-1}\,
        \sum_{z \in \mU_N} z^{j-k}
		\Tr(\sigma_{k-j}\cS_z)\\
		&=
        \frac{1}{N}
        \sum_{z \in \mU_N}  \Tr(\wh{\Sigma}_N(z)\cS_{z}),
	\end{align*}	
which establishes (\ref{equ:cost_fun:inf}).
\end{proof}	

\section{BIPARTITE ENTANGLEMENT CRITERION FOR THE NETWORK} \label{sec:Sep_TIN}

For any fixed but otherwise arbitrary nodes $j\ne k$ of the translation invariant network under consideration, we denote by
\begin{equation}
	\label{equ:xi}
	\Ups(t)
    :=
    {\small\begin{bmatrix}
        X_j(t)\\
        X_k(t)
        \end{bmatrix}}
\end{equation}
the augmented vector of dynamic variable of the corresponding bipartite system. In a stable network driven by external fields in the vacuum state, the mean values of the dynamic variables asymptotically vanish and will be considered to be zero in what follows.
Since the joint unitary evolution of the network and its environment preserves the CCRs (\ref{xCCR}), the quantum covariance matrix of the vector $\Ups$ in (\ref{equ:xi}) takes the form
	\begin{align}
\nonumber
		S(t)
        := &
        \bE(
            \Ups(t)
            \Ups(t)^\rT
        )
        =
        \Re S(t) +
        i
        {\small\begin{bmatrix}
            \Theta & 0\\
            0 & \Theta
        \end{bmatrix}}\\
\label{S}									
        = &
        {\small\begin{bmatrix}
            S_{11}(t) &  S_{12}(t) \\
			S_{21}(t) &  S_{22}(t)
        \end{bmatrix}}
        \succcurlyeq
        0,
\end{align}
where the off-diagonal blocks are real matrices satisfying $S_{12}^{\rT} = S_{21}$. Note that in view of (\ref{equ:cros_cov}), the matrix $S$
is representable in terms of the matrices (\ref{equ:DFT_cov}) whose evolution is described by Lemma~\ref{lem:cov_net_time}:
	\begin{align}
		\label{equ:Sigma_ij}
		S(t)
        & =
        \frac{1}{N^2}
        \sum_{z,v \in \mU_N}
        {\small\begin{bmatrix}
            (z/v)^j     & z^jv^{-k}\\
            z^kv^{-j}   & (z/v)^k
        \end{bmatrix}}
        \ox
        \cS_{z,v}(t).
	\end{align}	
Gaussian quantum states \cite{KRP_2010} are completely specified by the covariance matrix of the quantum variables and their mean vector, with the latter not affecting whether the state is separable or entangled.  Therefore, separability or entanglement of such states can be studied in terms of the covariance matrices \cite{Werner01}. Furthermore, the Gaussian nature of system states is inherited by subsystems and is preserved by linear QSDEs driven by vacuum boson fields   (provided the initial state is Gaussian). For what follows, we assume that the network consists of one-mode oscillators (that is, $n=1$) and is initialized at a Gaussian state. Then by applying the Peres-Horodecki-Simon criterion \cite{Sim2000}, it follows that the above described bipartite  system (formed from the $j$th and $k$th nodes of the network) is in a separable Gaussian state if and only if
\begin{align}
\nonumber
	\Lambda(t)
    :=  &
        \Re S(t)
        +
        i
        {\small\begin{bmatrix}
            \Theta  & 0\\
            0       & -\Theta
        \end{bmatrix}}\\
\label{equ:Lambda_chi_relation}
        = &
        S(t)
        +i
        {\small\begin{bmatrix}
            0 & 0\\
            0 & -2\Theta
        \end{bmatrix}}
        =
        {\small\begin{bmatrix}
            S_{11}(t)      & S_{12}(t) \\
            S_{12}(t)^\rT  & \overline{S_{22}(t)}
        \end{bmatrix}}
        \succcurlyeq 0,
\end{align}
where $S(t)$ is the quantum covariance $(4\x 4)$-matrix in (\ref{S}).
An equivalent alternative form of the separability criterion 
is provided by Lemma~\ref{lem:det_neg} of Appendix~\ref{app:fact_on_ent} which reduces  to checking the sign of $\det \Lambda(t)$ instead of verifying the $(4\x 4)$-matrix inequality in (\ref{equ:Lambda_chi_relation}). More precisely, according to Lemma~\ref{lem:det_neg}, the bipartite system is in an entangled Gaussian state if and only if the quantum covariance $(2\x 2)$-matrices $S_{11}$ and $S_{22}$ for the constituent nodes  $j$ and $k$ are positive definite and
$$		
    \det ( \overline{S_{22}} - S_{12}^\rT S_{11}^{-1} S_{12})< 0.
$$
 This is closely related to the following identity (see part (iii) of Appendix~\ref{app:fact_about_matrices} for details):
\begin{align*}	
	\det \Lambda(t)
    &=
    \det S_{11}
    \det ( \overline{S_{22}} - S_{12}^\rT S_{11}^{+} S_{12})	\\
    &=
    \det S_{22}
    \det ( S_{11} - S_{12}\overline{S_{22}}^{+} S_{12}^\rT ).
\end{align*}

\begin{thm}
\label{thm:entanglement}
Suppose the translation invariant network consists of $N$ one-mode open quantum harmonic oscillators and is stable in the sense of (\ref{stab}). Then the invariant Gaussian state of the bipartite system at the nodes $j\ne k$ is separable if and only if the infinite time horizon limit of the matrix $\Lambda(t)$ in (\ref{equ:Lambda_chi_relation}) satisfies
	\begin{align}
		\label{equ:cond_sep:1}
		\Lambda(\infty)
        :=
        \lim_{t\to +\infty}
        \Lambda(t)
        =
        \frac{1}{N}
        \sum_{z \in \mU_N}		
		{\small\begin{bmatrix}
            \cS_{z}             & z^{j-k} \cS_{z} \\
		  z^{k-j} \cS_{z}       & \overline{\cS_{z}}
        \end{bmatrix}}
		\succcurlyeq 0,
	\end{align}
	where $\cS_z$ is the unique solution of the algebraic Lyapunov equation (\ref{cSz}).
\end{thm}
\begin{proof}
In view of Lemma~\ref{lem:cov_net_time}, the fulfillment of the stability condition (\ref{stab}) implies (\ref{Szvinf}) whose combination  with (\ref{equ:Sigma_ij}) leads to
$$
		S(\infty)
        :=
        \lim_{t\to +\infty}
        S(t)
        =
        \frac{1}{N}
        \sum_{z \in \mU_N}		
		{\small\begin{bmatrix}
            \cS_{z}             & z^{j-k} \cS_{z} \\
		  z^{k-j} \cS_{z}       & \cS_{z}.
        \end{bmatrix}}
$$
The latter allows the separability criterion (\ref{equ:Lambda_chi_relation}) to be represented for the corresponding limit matrix $\Lambda(\infty)$ in the form (\ref{equ:cond_sep:1}).
\end{proof}

A closed-form solution of the algebraic Lyapunov equation (\ref{cSz}) for the one-mode case (which is considered in Theorem~\ref{thm:entanglement}) is presented in Appendix~\ref{app:Lya_sol}.

\section{INFINITE CHAIN OF LINEAR QUANTUM STOCHASTIC SYSTEMS} \label{sec:inf_chain}

We will now consider the infinite one-dimensional chain of directly coupled linear quantum stochastic systems whose finite fragments were introduced in Section~\ref{sec:Ring_Topology}.  A particular case of nearest neighbour interaction is depicted in Fig.~\ref{fig:chain_direct_coupling}.
\begin{figure}[htbp]
\begin{center}
	\begin{tikzpicture}[scale=.6]
		\node at (0,0) {$F$};
		\node at (-3,0) {$F$};
		\node at (3,0) {$F$};
		\node at (5.5,0) {...};
		\node at (-5.5,0) {...};

		\draw  (-1,-1) rectangle (1,1);
		\draw  (-4,-1) rectangle (-2,1);
		\draw  (2,-1) rectangle (4,1);

		\draw [<->, thick] (-1,0) -- (-2,0);
		\draw [<->, thick] (1,0) -- (2,0);
		\draw [<->, thick] (4,0) -- (5,0);
		\draw [<->, thick] (-4,0) -- (-5,0);

		\node at (-3,2.25) {\scriptsize $W_{k-1}$};
		\node at (0,2.25) {\scriptsize $W_k$};
		\node at (3,2.25) {\scriptsize $W_{k+1}$};

		\draw [->, thick] (3,2) -- (3,1);
		\draw [->, thick] (0,2) -- (0,1);
		\draw [->, thick] (-3,2) -- (-3,1);
\end{tikzpicture}
\end{center}
\caption{An infinite chain of directly coupled linear open quantum systems with nearest neighbour interaction. Also shown are the quantum noises (numbered from left to right).}
\label{fig:chain_direct_coupling}
\end{figure}
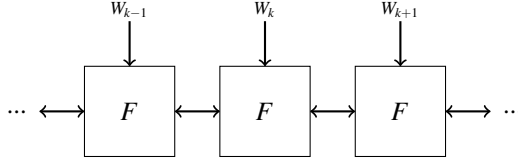
The results of Sections~\ref{sec:Stable_TIN} and \ref{sec:Sep_TIN}, which employ the representation of the translation invariant network in the spatial frequency domain (see Section~\ref{sec:Spatial_fourier_tran}),  can be extended to the infinite network case by letting the fragment length $N$ go to infinity.
The mean square cost functional $\cE_N$, discussed in Section~\ref{sec:Stable_TIN}, leads to a performance index for the infinite network. More precisely, similarly to \cite[Theorem 3]{VP2014}, under the stability condition (\ref{stab}), the infinite time horizon value $\cE_N(\infty)$ in (\ref{equ:cost_fun:inf}) has the following limit as the fragment length $N$ tends to infinity:
\begin{equation}
	\label{equ:cost_infch}
	\lim_{N \to+\infty}
    \cE_N(\infty)
    =
    \frac{1}{2 \pi i}
    \oint_\mU \Tr(\Sigma_z \cS_z)\frac{\rd z}{z}.
\end{equation}
Here, $\Sigma_z$ is the spectral density (\ref{Sigma_z}) of the absolutely summable weighting sequence. The resulting cost functional in (\ref{equ:cost_infch}) corresponds to the thermodynamic limit in equilibrium statistical mechanics \cite{Ruelle78}. The results of Theorem~\ref{thm:entanglement} on separability of states for bipartite subsystems (in the network of one-mode oscillators) can be extended to the infinite chain as follows.

\begin{thm}
	\label{thm:sep}
Suppose the infinite chain of one-mode open quantum harmonic oscillators satisfies the stability condition (\ref{stab}). Then the invariant Gaussian state of the bipartite system at the nodes $j\ne k$ is separable if and only if
	\begin{align}
		\label{equ:cond_sep_inf:1}
		\wh{\Lambda} :=
        \frac{1}{2\pi i}
        \oint_{z \in \mU}
		{\small\begin{bmatrix} \cS_{z} & z^{j-k} \cS_{z} \\
		z^{k-j} \cS_{z} & \overline{\cS_{z}}  \end{bmatrix}} \frac{\rd z}{z}
		\succcurlyeq 0,	
	\end{align}
	where $\cS_z$ is the unique solution of the algebraic Lyapunov equation (\ref{cSz}).
\end{thm}
\begin{proof}
Under the stability condition (\ref{stab}), the solution $\cS_z$ of (\ref{cSz}) is analytic (and hence, continuous) with respect to $z$  over a neighbourhood of the unit circle $\mU$. Hence, the matrix $\Lambda_N(\infty):= \Lambda(\infty)$ in (\ref{equ:cond_sep:1}), which is the Riemann sum of an appropriate integral, converges to the matrix $\wh{\Lambda}$ in (\ref{equ:cond_sep_inf:1}) as $N \to +\infty$.
\end{proof}

Under the assumptions of Theorem~\ref{thm:sep}, the matrix $\wh{\Lambda}$ in (\ref{equ:cond_sep_inf:1}) is representable as
\begin{equation}
\label{LamS}
    \wh{\Lambda}
    =
    {\small\begin{bmatrix}
        S_0 & S_{j-k}\\
        S_{j-k}^* & \overline{S_0}
    \end{bmatrix}}
\end{equation}
in terms of the Fourier coefficients of the spatial spectral density $\cS_z$:
	$$
		S_{\ell} := \frac{1}{2 \pi i} \oint_{z \in \mU} z^{\ell-1} \cS_z \rd z.
	$$
The latter describe the appropriate cross-covariance matrices for the invariant Gaussian state of the network and, in view of the Riemann-Lebesgue lemma,  satisfy $\lim_{\ell \to \infty}S_{\ell} = 0$. Therefore, if $S_0\succ 0$ (and hence, $\overline{S_0}\succ 0$, see Appendix~\ref{app:fact_about_matrices}), then (\ref{LamS}) implies that $\wh{\Lambda}\succ 0$ for all $|j-k|$ large enough. This means separability (that is, the absence of entanglement) for all nodes $j$ and $k$ which are sufficiently distant from each other.

\section{ILLUSTRATIVE EXAMPLE} \label{sec:example}

We have generated a sample of 100 sets of random coefficients of the QSDE (\ref{dXj}) for a translation invariant network of $N=400$ one-mode ($n=1$) oscillators subject to the stability condition of Lemma~\ref{lem:cov_net_time}. The interaction range in the network was $d=8$, and the dimension of the driving quantum noises was $2m= 2$. For each of these samples,  $\det \Lambda(\infty)$ from (\ref{equ:cond_sep:1}) and the log-negativity index \cite{VW02} were calculated for the bipartite systems at nodes $j\ne k$ of the network. The sample means and ranges of values of these entanglement measures versus $a:=j-k$ are depicted in Fig.~\ref{fig:Entg}. The mean values show that it is quite probable for the nodes of the network to be entangled within the interaction range (that is, for $|a|\< 8$ in this example). According to Fig.~\ref{fig:Entg}, the entanglement vanishes for sufficiently distant nodes (with $|a|> 8$ in this case).
\begin{center}
\begin{figure}[htbp]
	\hspace*{3cm}
    \includegraphics[width=12cm,height=7cm]{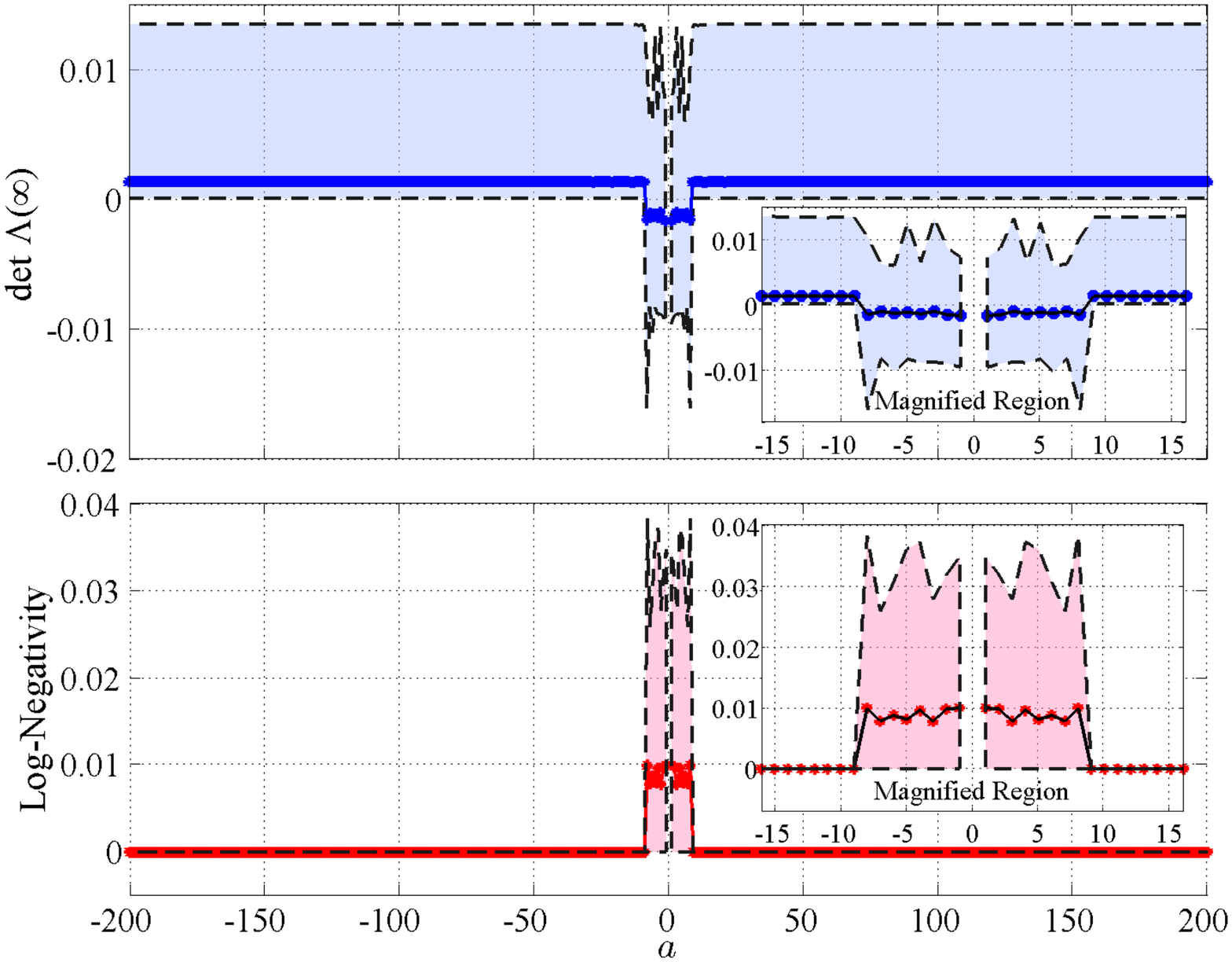}
    \caption{Mean values (solid lines) and ranges of values (colour filled domains) for $\det \Lambda(\infty)$ and the log-negativity versus $a:=j-k$ for nodes $j$ and $k$ in 100 randomly generated stable translation invariant networks of 400 one-mode oscillators. Negative values of $\det \Lambda(\infty)$ and the  corresponding positive values of the log-negativity indicate the quantum entanglement between the nodes.}\label{fig:Entg}
\end{figure}
\end{center}

\section{CONCLUSION}\label{sec:conc}

We have presented the modelling of translation invariant networks of directly coupled linear quantum stochastic systems with finite-range interaction. Using the spatial Fourier transforms and LMIs, we have provided conditions of stability  for the network, and considered a mean square performance functional with block Toeplitz weights for the stable network. For networks of one-mode oscillators in a Gaussian state, we represented the Peres-Horodecki-Simon entanglement criterion (for bipartite systems at arbitrary nodes) in the spatial frequency domain. We have discussed an extension of the results on stability and entanglement to infinite chains of linear quantum systems by letting the network size go to infinity. We have also provided a numerical example which revealed the presence of entanglement between the nodes of the network within the interaction range.

\appendix
\renewcommand{\theequation}{\Alph{subsection}\arabic{equation}}

\subsection{Solution of the Algebraic Lyapunov Equation}\label{app:Lya_sol}
\setcounter{equation}{0}

The solution $\cS_z$ of the algebraic Lyapunov equation (\ref{cSz})
can be computed by using the column-wise vectorization  $\vec(\cdot)$ of matrices \cite{M_1988} as
\begin{equation}
	\label{equ:vecSz}
	\vec (\cS_z)= -\Gamma \vec (B \Omega B^\rT),
    \qquad
    \Gamma:=(\overline{\cA_z} \oplus \cA_z)^{-1},
\end{equation}
where the dependence of $\Gamma$ on $z$ is omitted for brevity, and $\overline{\cA_z} \oplus \cA_z := I_{2n} \ox \cA_z + \overline{\cA_z} \ox I_{2n}$ is the Kronecker sum of $\overline{\cA_z}$ and $\cA_z$. Since $\overline{\cA_z} = \cA_{1/z}$ for any $z\in \mU$ in view of (\ref{equ:A_z}), it follows from (\ref{equ:vecSz}) that $\cS_z$ is a rational function of $z$ which is analytic in a neighbourhood of the unit circle $\mU$, provided the stability condition (\ref{stab}) is satisfied. Indeed, the latter guarantees that the matrix $\overline{\cA_z} \oplus \cA_z$ is also Hurwitz (and hence, nonsingular) for any $z \in \mU$. Similarly to the approach in \cite{BS1970}, the matrix  $\Gamma$ can be computed as
\begin{align}	
	\nonumber
	\Gamma = &  (I_{2n} \ox D^{-1})(I_{2n} \ox \cA_z^{2n-1} - E_1 \ox \cA_z^{2n-2}\\
	\label{equ:Gamma_n}	
	&+E_2 \ox \cA_z^{2n-3}-  \ldots -(-1)^{2n}E_{2n-1} \ox I_{2n}),	\\		
	\label{equ:D_n}
		 D:= & \cA_z^{2n}-c_1\cA_z^{2n-1}+c_2\cA_z^{2n-2} -\ldots +c_{2n}I_{2n},\\
	\label{equ:E1_n}		
		 E_1:= & \overline{\cA_z}+c_1I_{2n},\\
	\label{equ:Ek_n}
		 E_k:=&\overline{\cA_z} E_{k-1} + c_k I_{2n}, \quad k=2, \ldots, 2n-1
\end{align}
where $c_j$ are the coefficients of the characteristic polynomial of $\overline{\cA_z}$, that is
\begin{equation*}
	\det(\lambda I_{2n}-\overline{\cA_z}) = \lambda^{2n}+c_1\lambda^{2n-1}+\ldots+c_{2n}.
\end{equation*}
In the case when the nodes of the translation invariant network are one-mode quantum harmonic oscillators (that is, $n=1$), the representation (\ref{equ:Gamma_n})--(\ref{equ:Ek_n}) reduces to
\begin{align}
	\nonumber
	\Gamma &= (I_{2} \ox D^{-1})(I_{2} \ox \cA_z - E_1 \ox I_2),	\\	\label{equ:Gamma_1}	
	  	   &= I_{2} \ox (D^{-1} \cA_z) - E_1 \ox (D^{-1}),	\\ \label{DD}
		 D& = \cA_z^2+ \cA_z \Tr \overline{\cA_z} +I_2\det \overline{\cA_z},\\ \nonumber		
		 E_1& = \overline{\cA_z}-I_{2} \Tr \overline{\cA_z}.
\end{align}
Furthermore, in this case, the Caley-Hamilton theorem leads to
$
	\cA_z^2=\cA_z\Tr \cA_z - I_2 \det \cA_z
$,
and hence, the matrix $D$ in (\ref{DD}) takes the form
\begin{equation*}
	D=\Tr(\cA_z+\cA_z^*)\cA_z+(\det\overline{\cA_z}-\det \cA_z)I_2.
\end{equation*}
In order to obtain a closed-form expression for $\cS_z$, the matrix $D^{-1}$ can be computed as
\begin{equation*}
	D^{-1}=\frac{1}{\det D}(I_2 \Tr D-D).
\end{equation*}
Here,
\begin{align*}
	\det D&= \Tr(\cA_z+\cA_z^*)^2 \det( \lambda_z I_2 - \cA_z )\\
			&= \Tr(\cA_z+\cA_z^*)^2 (\lambda_z^2 - \Tr(\cA_z)\lambda_z + \det(\cA_z)),\\
			&= 4(\Re v_z)^2(\lambda_z^2 - v_z\lambda_z + u_z ),
\end{align*}
where $v_z:=\Tr \cA_z$, $u_z:=\det\cA_z$, and $\lambda_z:= \frac{\Im u_z}{\Re v_z}$. Therefore,
\begin{align*}
	D^{-1}&= \frac{\Re v_z((v_z-2\Im u_z)I_2-\cA_z)}
	{(\Im u_z)^2-2v_z\Re v_z\Im u_z +2(\Re v_z)^2u_z}.
\end{align*}
This implies that
\begin{equation*}
	\label{equ:S_z_temp}
	\cS_z = -D^{-1}(\cA_zB\Omega B^\rT-B\Omega B^\rT E_1^\rT),	
\end{equation*}
where use is made of (\ref{equ:vecSz}), (\ref{equ:Gamma_1}) and matricization techniques.

\subsection{An Equivalent Bipartite Entanglement Criterion for Two-mode Quantum Harmonic Oscillators in Gaussian States}
\label{app:fact_on_ent}
\setcounter{equation}{0}

\begin{lem}
	\label{lem:det_neg}
	Suppose a two-mode quantum harmonic oscillator with the CCR matrix $I_2\ox \Theta$ is in a Gaussian state with the quantum covariance $(4\x 4)$-matrix $S$ in (\ref{S}). Then this state is entangled if and only if the matrix $\Lambda$ in (\ref{equ:Lambda_chi_relation}) satisfies
	\begin{equation}
		\label{equ:det_negativity}
		\det \Lambda <0.
	\end{equation}
\end{lem}
\begin{proof}
Positive semi-definiteness of a Hermitian matrix is equivalent to nonnegativeness of all its principal minors \cite[p. 421]{bernstein2009}. 	
In view of (\ref{S}) and (\ref{equ:Lambda_chi_relation}), all the principal minors of $\Lambda$ up to order $3$ inherit nonnegativeness from those of the quantum covariance matrix $S\succcurlyeq 0$ (see part (iii) of Appendix~\ref{app:fact_about_matrices} for details). Therefore, violation of the property $\Lambda \succcurlyeq 0$ is equivalent to  (\ref{equ:det_negativity}), thus providing an equivalent form of the Peres-Horodecki-Simon criterion \cite{Sim2000} of entanglement.
\end{proof}

\subsection{Several Auxiliary Facts on Hermitian Matrices}
\label{app:fact_about_matrices}
\setcounter{equation}{0}

For ease of reference, we mention the following properties of Hermitian matrices
$S=\Re S + i\Im S$:
\begin{enumerate}
	\item[i)]
	$S\succcurlyeq 0$ if and only if
	\begin{equation}
		\label{equ:cmplx_mat_sgn}
			{\small \begin{bmatrix}
				\Re S & -\Im S \\
				\Im S & \Re S
			\end{bmatrix}} \succcurlyeq 0.
	\end{equation}
	\item[ii)]
	$S\succcurlyeq 0$ if and only if $\overline{S}\succcurlyeq 0$.
	\item[iii)]
	Suppose $S\succcurlyeq 0$ is a $(4 \x 4)$-matrix with
	$
	\Im S := I_2 \ox \Theta    	
	$, and $\Theta \in \mA_2$.
	Then all the principal minors of the matrix
	\begin{equation}
    \label{Lambda}
		\Lambda := \Re S + i {\small\begin{bmatrix}
						\Theta & 0\\ 0 & -\Theta
			   	\end{bmatrix}}
			   	= {\small\begin{bmatrix}
					\Lambda_{11} & \Lambda_{12} \\ \Lambda_{12}^\rT & \Lambda_{22}
				  \end{bmatrix}}		
	\end{equation}
	are nonnegative, except probably $\det \Lambda$. The latter determinant can be computed in terms of the $(2\x 2)$-blocks of $\Lambda$ as
			  \begin{align}
            \nonumber
					\det \Lambda & = \det \Lambda_{11} \det ( \Lambda_{22} - \Lambda_{12}^\rT
					\Lambda_{11}^{+} \Lambda_{12})\\
			 \label{equ:det}
					& = \det \Lambda_{22} \det ( \Lambda_{11} - \Lambda_{12}
					\Lambda_{22}^{+} \Lambda_{12}^\rT).
			  \end{align}
\end{enumerate}
\begin{proof}
(i) is verified by inspection.
In view of (i), (ii) can be proved by left and right multiplying both sides of (\ref{equ:cmplx_mat_sgn}) by  ${\scriptsize\begin{bmatrix}
	I & 0 \\ 0 & -I
	\end{bmatrix}}$. To prove (iii), note that the matrix $\Lambda$ in (\ref{Lambda}) shares with $S$ or $\overline{S}$ all principal submatrices up to order $3$. These submatrices inherit positive semi-definiteness from $S$ and $\overline{S}$, and hence, the corresponding minors of $\Lambda$ are nonnegative. Therefore, only the fourth order minor $\det \Lambda$ can be negative. Now, elementary matrix transformations lead to
			  \begin{align}
			  		\nonumber
			  		{\small\begin{bmatrix}
			  			I                               & 0\\
                        -\Lambda_{12}^\rT\Lambda_{11}^+ & I
			  		\end{bmatrix}}	&
					{\small\begin{bmatrix}
			  			\Lambda_{11}            & \Lambda_{12}\\
                        \Lambda_{12}^\rT        & \Lambda_{22}
			  		\end{bmatrix}}			
			  		{\small\begin{bmatrix}
			  			I & -\Lambda_{11}^+\Lambda_{12}\\
                        0 & I
			  		\end{bmatrix}}	\\ \nonumber
			  		&=
			  		{\small\begin{bmatrix}
			  			\Lambda_{11}                            & \Lambda_{12}-\Lambda_{11} \Lambda_{11}^+ \Lambda_{12} \\
			  			 \Lambda_{12}^\rT-\Lambda_{12}^\rT \Lambda_{11}^+\Lambda_{11}  &
                        								 \Lambda_{22}
								 -\Lambda_{12}^\rT
								 \Lambda_{11}^+ \Lambda_{12}		  	
			  		\end{bmatrix}}\\ \label{equ:app:diag}
			  		&=	
			  		{\small\begin{bmatrix}
			  			\Lambda_{11} & 0\\
			  			 0 &
								 \Lambda_{22}
								 -\Lambda_{12}^\rT
								 \Lambda_{11}^+ \Lambda_{12}		 			  	
			  		\end{bmatrix}}.
			 \end{align}			
			 Here, use is made of the properties that $\Lambda_{11}^+ \Lambda_{11} \Lambda_{11}^+ =  \Lambda_{11}^+$ and $\Lambda = \Lambda^*$ together with
			 \begin{align}
			 	\label{equ:temp_AM}
			 	(\Lambda_{11} \Lambda_{11}^+  -I )\Lambda_{12}
			 	= (S_{11} S_{11}^+  -I )
			 	S_{12}=0,
			 \end{align}
			 where the latter follows from $\Lambda_{11} = S_{11}$, $\Lambda_{12} = S_{12}$ and the assumption that $S\succcurlyeq 0$ \cite{zhang2006}. 
Indeed, every $S \succcurlyeq 0$ has a unique square root $\sqrt{S}\succcurlyeq 0$ whose
columns can be partitioned according to the blocks of $S$. More precisely, let $\sqrt{S} := {\begin{bmatrix} L_1 & L_2 \end{bmatrix}}$, and hence,
	\begin{equation}
    \label{SL}
		S =
    {\small\begin{bmatrix}
					S_{11} & S_{12}\\
					S_{21} & S_{22}
				\end{bmatrix}}
    =
    {\small\begin{bmatrix}
					L_1^* L_1 & L_1^* L_2\\
					L_2^* L_1 & L_2^* L_2
				\end{bmatrix}}.
	\end{equation}
	Consider the polar decomposition of $L_1= \Phi \Xi$, where $\Phi$ has orthonormal columns and $\Xi= \Xi^*\succcurlyeq 0$. Then $S_{11}= L_1^* L_1 =\Xi^* \Xi=\Xi^2$ and $S_{12}=  L_1^* L_2 = \Xi \Phi^* L_2$ in view of (\ref{SL}). This implies $S_{11} S_{11}^+S_{12} = \Xi^2 (\Xi^2)^+ \Xi \Phi^* L_2= \Xi \Xi^+ \Xi \Phi^* L_2=S_{12}$, thus establishing (\ref{equ:temp_AM}). Here, use is also made of $\Xi^+= \Xi^*(\Xi\Xi^*)^+$, the fact that $\Xi$ is Hermitian and the definition of the generalized inverse.
Now, by evaluating the determinant of both sides of (\ref{equ:app:diag}), it follows that $\det \Lambda$ is indeed representable by the first of the equalities  (\ref{equ:det}). The second of them is established in a similar fashion.
\end{proof}

\begin{thebibliography}{99}{\scriptsize
\bibitem{BS1970}
  S.Barnett, and C.Storey,
  \emph{Matrix Methods in Stability Theory}, Nelson, 1970.
\bibitem{bernstein2009}
D.S.Bernstein,
\emph{Matrix Mathematics: Theory, Facts, And Formulas},
Princeton University Press, 2009.
\bibitem{EB_2005}
S.C.Edwards, and V.P.Belavkin,
Optimal quantum filtering and
quantum feedback control,
\emph{arXiv:quant-ph/0506018v2}, August 1,  2005.
\bibitem{GZ_2004}
C.W.Gardiner, and P.Zoller,
\textit{Quantum Noise},
Springer, Berlin, 2004.
\bibitem{GKS_1976}
V.Gorini, A.Kossakowski, and E.C.G.Sudarshan, Completely positive dynamical
semigroups of N-level systems, \emph{J. Math. Phys.}, vol. 17, no.
5, 1976, pp. 821--825.
\bibitem{gren58}
  U.Grenander, and G.Szeg\"{o}, \textit{Toeplitz Forms and Their Applications},
  University of California Press, Berkeley, 1958.
\bibitem{hormander73}
  L.H\"{o}rmander,
  \emph{An Introduction to Complex Analysis in Several Variables},
  Elsevier,
  1973.
\bibitem{Ho2009}
R.Horodecki, P.Horodecki, M.Horodecki, and K.Horodecki, Quantum entanglement, \emph{Rev. Mod. Phys.}, vol. 81, no. 2, 2009, pp. 865--942.
\bibitem{HP_1984}
R.L.Hudson,  and K.R.Parthasarathy,
Quantum It\^{o}'s formula and stochastic evolutions,
\emph{Commun. Math. Phys.}, vol.  93, 1984, pp. 301--323.
\bibitem{L_1976}
G.Lindblad, On the generators of quantum dynamical semigroups,
\emph{Commun. Math. Phys.}, vol. 48, 1976, pp. 119--130.
\bibitem{M_1988}
J.R.Magnus,
\emph{Linear Structures},
Oxford University Press, New York, 1988.
\bibitem{Newman99}
  M.E.J.Newman and G.T.Barkema,
  \emph{Monte Carlo Methods in Statistical Physics},
  Clarendon Press Oxford, 1999.
\bibitem{NC_2000}
M.A.Nielsen, and I.L.Chuang,
\textit{Quantum Computation and Quantum Information},
Cambridge University Press, Cambridge, 2000.
\bibitem{Notomi2008}
M.Notomi, E.Kuramochi,  and T.Tanabe,
Large-scale arrays of ultrahigh-Q coupled nanocavities,
\emph{Nature Photonics}, vol. 2, no. 12, 2008, pp. 741--747.
\bibitem{P_1992}
K.R.Parthasarathy,
\emph{An Introduction to Quantum Stochastic Calculus},
Birk\-h\"{a}user, Basel, 1992.
\bibitem{KRP_2010}
K.R.Parthasarathy,
What is a Gaussian state?
\emph{Commun. Stoch. Anal.}, vol. 4, 2010, pp. 143--160.
\bibitem{Peres96}
A.Peres,
Separability criterion for density matrices,
\emph{Phys. Rev. Lett.}, vol. 77, no. 8, 1996, pp. 1413--1415.
\bibitem{Ian2012robust}
I.R.Petersen, V.Ugrinovskii, and M.R.James
Robust stability of uncertain linear quantum systems,
\emph{Phil. Trans. Royal Soc. A}, vol. 370, no. 1979, 2012, pp. 5354--5363.
\bibitem{Quach11}
J.Q.Quach, C.-H.Su, A.M.Martin, A.D.Greentree, and L.C.L.Hollenberg,
  Reconfigurable quantum metamaterials,
  \emph{Optics Express},
  vol. 19,
  no. 12,
  2011,
  pp. 11018--11033.
\bibitem{Rakhmanov2008}
  A.L.Rakhmanov, A.M.Zagoskin, S.Savel'ev, and F.Nori,
  Quantum metamaterials: Electromagnetic waves in a Josephson qubit line,
  \emph{Phys. Rev. B},
  vol.77,
  no. 14,
  2008,
  p. 144507.
 \bibitem{Ruelle78}
  D.Ruelle,
  \emph{Thermodynamic Formalism}, Addison-Wesley, London, 1978.
\bibitem{Serafini05}
  A.Serafini, M.G.A.Paris, F.Illuminati, and S.De.Siena,
  Quantifying decoherence in continuous variable systems,
  \emph{J. Optics B: Quantum and Semiclassical Optics},
  vol. 7,
  no. 4,
  2005,
  p. R19.
\bibitem{SVP_2015}
A.K.Sichani, I.G.Vladimirov, and I.R.Petersen,
Robust mean square stability of open quantum stochastic systems with Hamiltonian perturbations in a Weyl quantization form,
Proc. Australian Control Conference,
Canberra, 17-18 November 2014, pp. 83--88 (arXiv:1503.02122 [quant-ph], 7 March 2015).
\bibitem{Sim2000}
R.Simon,
Peres-Horodecki Separability Criterion for Continuous Variable Systems,
\emph{Phys. Rev. Lett.}, vol. 84, 2000, pp. 2726--2729.
\bibitem{Veselago2006}
V.Veselago, L.Braginsky, V.Shklover, and C.Hafner,
Negative refractive index materials,
\emph{J. Comput. Theor. Nanoscience},
vol. 3, no. 2, 2006, pp. 189--218.
\bibitem{VW02}
G.Vidal, and R.F.Werner,
Computable measure of entanglement,
\emph{Phys. Rev. A},
vol. 65,
no. 3,
2002,
pp. 032314.
\bibitem{VP2014}
I.G.Vladimirov, I.R.Petersen,
Physical realizability and mean square performance of translation invariant networks of interacting linear quantum stochastic systems,
Proc. 21st Int. Symp. on Math. Theory of Networks and Systems, Groningen, Netherlands, June 7--11,  2014,  pp. 1881--1888.
\bibitem{Werner01}
  R.F.Werner, and M.M.Wolf,
  Bound entangled Gaussian states,
  \emph{Phys. Rev. Lett.},
  vol. 86,
  no. 16,
  2001,
  pp. 3658--3661.
\bibitem{Zagoskin11}
  A.M.Zagoskin,
  \textit{Quantum Engineering: Theory and Design of Quantum Coherent Structures}, Cambridge
  University Press, 2011.
\bibitem{Zagoskin12}
A.M.Zagoskin,
Superconducting quantum metamaterials in 3D: possible realizations,
\emph{J. Optics}, vol. 14, 2012, p. 114011.
 \bibitem{zhang2006}
  F.Zhang, \emph{The Schur Complement and Its Applications}, Springer, 2005.
\bibitem{Zheludev11}
N.I.Zheludev,
A roadmap for metamaterials,
\emph{Opt. Photon. News}, vol. 22, 2011, pp. 30--35.
}
\end{thebibliography}
\end{document}